\documentclass[letterpaper, 10 pt, conference]{ieeeconf}  \renewcommand{\baselinestretch}{0.97}

\IEEEoverridecommandlockouts                             
\usepackage{float}
\usepackage{verbatim}
\usepackage{paralist}
\usepackage{amsfonts}
\usepackage{graphics} 
\usepackage{epsfig} 
\usepackage{mathptmx} 
\usepackage{times} 
\usepackage{amsmath} 
\usepackage{amssymb}  
\usepackage{physics}
\usepackage{mdwmath}
\usepackage{mdwtab}
\usepackage{color}
\usepackage[hidelinks]{hyperref}
\usepackage{algpseudocode}
\usepackage{algorithm}
\usepackage{caption}
\usepackage{subcaption}
\usepackage{amsmath} 
\usepackage{bm}      

\newtheorem{prop}{Proposition}
\newtheorem{remark}{Remark}

\title{\LARGE \bf
State-Constrained Optimal Control for Coherence Preservation in Multi-Level Open Quantum Systems
}

\author{Nahid Binandeh Dehaghani, A. Pedro Aguiar, Rafal Wisniewski
\thanks{N. Dehaghani and A. Aguiar are with the Research Center for Systems and Technologies (SYSTEC), Electrical and Computer Engineering Department, FEUP - Faculty of Engineering, University of Porto, Rua Dr. Roberto Frias sn, i219, 4200-465 Porto, Portugal
        {\tt\small \{nahid,pedro.aguiar\}@fe.up.pt}}%
\thanks{R. Wisniewski is with Department of Electronic Systems, Aalborg University, Fredrik Bajers vej 7c, DK-9220 Aalborg, Denmark
        {\tt\small raf@es.aau.dk}}%
\thanks{The authors acknowledge the support of FCT for the grant 2021.07608.BD, the ARISE Associated Laboratory, Ref. LA/P/0112/2020, and the R$\&$D Unit SYSTEC-Base, Ref. UIDB/00147/2020, and Programmatic, Ref. UIDP/00147/2020 funds, and also the support of projects SNAP, Ref. NORTE-01-0145-FEDER-000085, and  RELIABLE (PTDC/EEI-AUT/3522/2020) funded by national funds through FCT/MCTES. The work has been done in the honor and memory of Professor Fernando Lobo Pereira.}
}
\begin{document}
\maketitle
\thispagestyle{empty}
\pagestyle{empty}

\begin{abstract} 
This paper addresses the optimal control of quantum coherence in multi-level systems, modeled by the Lindblad master equation, which captures both unitary evolution and environmental dissipation. We develop an energy minimization framework to control the evolution of a qutrit (three-level) quantum system while preserving coherence between states. The control problem is formulated using Pontryagin’s Minimum Principle in the form of Gamkrelidze, incorporating state constraints to ensure coherence remains within desired bounds. Our approach accounts for Markovian decoherence, demonstrating that the Lindblad operator is non-unital, which reflects the irreversible decay processes inherent in the system. The results provide insights into effectively maintaining quantum coherence in the presence of dissipation.
\end{abstract} 

\section{Introduction}
Quantum coherence is a crucial phenomenon in quantum systems, underlying many key applications in quantum computing, quantum information processing, and quantum control \cite{yuan2022preserving}. However, decoherence — the loss of coherence due to interaction with the environment — poses a significant challenge to the preservation of quantum superpositions and entanglement. For multi-level quantum systems, such as qutrits (three-level systems), decoherence processes can lead to the rapid degradation of coherence and the associated quantum advantages. One of the most commonly encountered types of decoherence is Markovian decoherence, where the system interacts with its environment in a memoryless manner, leading to a continuous and irreversible loss of coherence. This non-unitary evolution can be effectively described using the Lindblad master equation \cite{manzano2020short}.

Markovian decoherence, particularly in systems characterized by non-unital decoherence channels, poses a distinct challenge compared to unital channels. 
Non-unital channels, such as those in spontaneous emission, result in irreversible decay to ground states, leading to system degradation \cite{blanchard2006decoherence}. The preservation of coherence in such systems is crucial, as coherence between quantum states is responsible for the interference patterns and superpositions that form the basis of quantum technology \cite[chapter 12]{cong2014control}.

To counteract the effects of decoherence, several strategies have been proposed, including quantum error-correction codes \cite{Josu2021}, decoherence-free subspaces \cite{hu2021optimizing}, dynamical decoupling techniques \cite{yuan2022preserving}, and classical feedback \cite{yuan2022preserving}. 
 
An alternative approach is to use optimal control theory to design control fields that maintain or restore quantum coherence during the system's evolution. Reference \cite{cui2008optimal} employs Pontryagin's Maximum Principle to achieve optimal decoherence control, focusing on how environmental engineering—particularly the coupling and state of the reservoir—can impact decoherence in a non-Markovian, open, dissipative quantum system. This approach, applied to a two-level system, highlights the significant influence that optimal control can have on decoherence dynamics.

This work complements our recent work in \cite{dehaghani2022quantum}, where we had proposed a quantum optimal control problem with state constraints in order to maximize fidelity in a quantum state transfer problem. Here, we focus on the problem of preserving the coherence of a multi-level quantum system subjected to Markovian decoherence in a minimum energy optimal control problem. We explore the dynamics of the system's density matrix, which governs both the population of states and the off-diagonal coherence elements. Using the Lindblad formalism, we model the non-unitary evolution of a multi-level quantum system and investigate how control fields can be applied to counteract decoherence. 

The main contributions of this work are twofold: (1) we develop a minimum energy optimal control strategy while preserving coherence, and (2) we implement state constraints to ensure the system's coherence remains within specified bounds throughout the evolution. By applying Pontryagin's Minimum Principle in Gamkrelidze form, we derive the necessary conditions for the optimal control, taking into account both the Hamiltonian evolution and the dissipative effects of the environment. Numerical simulations, given for a qutrit system with three energy levels, demonstrate the feasibility of this approach in maintaining coherence, providing insight into the control mechanisms necessary to mitigate decoherence in realistic quantum systems.


The paper is organized as follows: Section II discusses the theoretical background on density matrix dynamics and coherence in multi-level systems, with a focus on quantum coherence. Section III presents the optimal control problem statement. Section IV introduces Pontryagin’s Minimum Principle in its Gamkrelidze form. In Section V, the application of the developed algorithm to a qutrit system is detailed, including numerical simulations. Finally, Section VI provides the conclusions and outlines future directions.

\section{Density Matrix Dynamics and Coherence in Multi-Level Systems}

The state of an $N$-level quantum system is represented by a density matrix $\rho$ of size $N \times N$, where $N$ is the number of quantum levels. The general form of the density matrix $\rho$ is
\[
\rho = \begin{pmatrix}
\rho_{11} & \rho_{12} & \dots & \rho_{1N} \\
\rho_{21} & \rho_{22} & \dots & \rho_{2N} \\
\vdots & \vdots & \ddots & \vdots \\
\rho_{N1} & \rho_{N2} & \dots & \rho_{NN}
\end{pmatrix}.
\]
The set of density operators for an $N$-level quantum system, denoted by $\mathcal{D} \subset \mathbb{C}^{N \times N}$, is the set of all matrices $\rho$ that satisfy
\[
\mathcal{D} = \left\{ \rho \in \mathbb{C}^{N \times N} : \rho = \rho^\dagger, \rho \geq 0, \text{Tr}(\rho) = 1 \right\}.
\]
These properties ensure that $\rho$ represents a valid quantum state, whether pure or mixed, and allow for the evolution of the system while maintaining the correct physical constraints.

The diagonal elements $\rho_{jj}$ represent the populations of the system in the corresponding states, while the off-diagonal elements $\rho_{jk}$ for $j \neq k$ represent the coherence between states $j$ and $k$. These off-diagonal elements are crucial for describing quantum coherence, as they reflect the superpositions between different quantum states. When these elements are non-zero, the system exhibits quantum coherence. If they vanish, the system becomes a classical mixture of states. Therefore, preserving or manipulating these coherences is essential in quantum control.

The time evolution of the density matrix $\rho$ for an open quantum system is governed by the Lindblad master equation, which takes into account both the coherent dynamics of the system (described by the Hamiltonian $H$) and the dissipative processes (described by the Lindblad operators $\mathcal{L}$), which represent the interaction of the system with its environment. 

In this case, the total Hamiltonian $H$ consists of the \textit{free Hamiltonian} $H_0$ that describes the intrinsic energy levels and interactions within the system, and a term that depends explicitly on the input signal $u(t)$, such that
\begin{equation}\label{H}
 H(t) = H_0 + H_C(t)
\end{equation}
where
$H_C = H_c(t) u(t)$ represents the time-dependent control Hamiltonian.

The general form of the Lindblad master equation is given by \cite{dehaghani2022quantum}
\begin{equation}\label{lindblad}
\dot \rho = -\frac{i}{\hbar} [H(t), \rho] + \mathcal{L}(\rho).
\end{equation}
Here, $\hbar$ represents the reduced Planck constant, which governs the relation between energy and frequency in quantum mechanics.
The first term describes the unitary evolution due to the total Hamiltonian $H(t)$, where $[H(t), \rho]$ is the commutator of $H(t)$ and $\rho$. The second term, $\mathcal{L}(\rho)$, accounts for the non-unitary evolution resulting from the interaction of the system with its environment, leading to decoherence and dissipation. The Lindblad operator $\mathcal{L}(\rho)$ for a multi-level system can be written as
\begin{equation}\label{lindbladop}
\mathcal{L}(\rho) = \sum_k \gamma_k \left( L_k \rho L_k^\dagger - \frac{1}{2} \left\{ L_k^\dagger L_k, \rho \right\} \right),
\end{equation}
where $\gamma_k$ are the decay rates associated with different decoherence channels, and $L_k$ are the Lindblad operators that describe the specific interactions between the system and the environment. These dissipative effects typically degrade the coherence in the system by causing the off-diagonal elements of the density matrix to decay over time.
Lindblad superoperators can be categorized as either unital or non-unital. A superoperator \(\mathcal{L}\) is considered unital if it satisfies \(\mathcal{L}(I) = 0\), implying that the identity operator (representing a completely mixed state) remains unchanged under its action. This condition holds true if and only if all Lindblad operators \(L_k\) are Hermitian, as demonstrated by the relation \(\mathcal{L}(I) = \sum_k [L_k, L_k^\dagger]\) \cite{rooney2012control}.

\subsection{Quantum Coherence in Multi-Level Systems}

Quantum coherence in multi-level systems is a generalization of the superposition concept from two-level systems (qubits) to systems with $N$ energy levels. In these systems, coherence is distributed across multiple off-diagonal elements of the density matrix, representing superpositions between different quantum states. 


To further analyze the coherence between specific pairs of states, we define operators that capture the real and imaginary parts of the coherence. For two states $|j\rangle$ and $|k\rangle$, these operators are described as
\begin{align*}
\delta_{jk}^{\text{Re}} &= |j\rangle \langle k| + |k\rangle \langle j|, \\    
\delta_{jk}^{\text{Im}} &= -i |j\rangle \langle k| + i |k\rangle \langle j|.
\end{align*}
The operator $\delta_{jk}^{\text{Re}}$ describes the real part of the coherence, which captures the population coherence between the two states, while $\delta_{jk}^{\text{Im}}$ captures the imaginary part, which encodes the phase coherence. Both components are essential for understanding quantum interference and superposition phenomena in multi-level systems.

The expectation values of these operators with respect to the density matrix $\rho$ provide the real and imaginary contributions to the coherence:
\begin{align*}
\langle \delta_{jk}^{\text{Re}} \rangle_\rho &= \text{Tr}(\rho \delta_{jk}^{\text{Re}}), \\
\langle \delta_{jk}^{\text{Im}} \rangle_\rho &= \Tr(\rho \delta_{jk}^{\text{Im}}).
\end{align*}

A global measure of coherence for an $N$-level system can be constructed by summing the contributions from the real and imaginary parts of the coherence between all pairs of states. The coherence function $C(\rho)$ is defined as \cite{cong2014control}
\begin{equation} \label{eq:C}
C(\rho) = \sqrt{\sum_{j < k} \left( \langle \delta_{jk}^{\text{Re}} \rangle_\rho \right)^2 + \left( \langle \delta_{jk}^{\text{Im}} \rangle_\rho \right)^2}.   
\end{equation}
This function provides a comprehensive measure of the system's coherence by considering both the real and imaginary parts of the coherence for every pair of states $(j, k)$. It captures the overall coherence present in the system and serves as an indicator of how well quantum superpositions are maintained as the system evolves according to the Lindblad equation. 

\section{Optimal Control Problem Statement}

This section outlines the formulation of the quantum optimal control problem. The methodology proposed here for designing an optimal quantum controller hinges on several key factors: the selection of the cost functional, the formulation of the Pontryagin-Hamiltonian, and the computational approach used to solve the optimality conditions defined by Pontryagin’s Minimum Principle (PMP). 
In quantum control, achieving minimal energy expenditure while preserving system coherence presents a significant challenge. To address this, we focus on an energy minimization approach that incorporates state constraints specifically aimed at maintaining quantum coherence within the control problem. Specifically, we consider the following cost function:
\begin{equation*}
J = 
\int_{{{t}_{0}}}^{t_f} u^2(t) \, dt,
\end{equation*}
which represents the total energy of the control field over the (fixed) time interval \([t_0, t_f]\). This cost functional is a standard choice in quantum and molecular control problems. Here, the control signal \(u(t):[t_0, t_f] \to \mathcal{U} := \{ u \in L_{\infty} : u(t) \in \Omega \subset \mathbb{R} \}\) is assumed to be bounded and measurable.
The optimal control problem with state constraints denoted as \((P)\), is formulated as follows:
\begin{equation*}
(P)\left\{ \begin{aligned}
  &\qquad \min_{u(\cdot)}  J
  \\ 
  & \text{subject to}  \\ 
  & \dot{\rho} = -\frac{i}{\hbar} [H, \rho] + \mathcal{L}(\rho), \quad \text{for almost every } t \in [t_0, t_f] \\
  & \rho(t_0) = \rho_0 \in \mathbb{C}^{N \times N} \\ 
  & u(t) \in \mathcal{U} 
  \quad \text{for almost every } t \in [t_0, t_f] \\
  & \alpha \leq {C}^2(\rho) \leq \beta, \quad \text{for all } t \in [t_0, t_f] \\
\end{aligned} \right.
\end{equation*}


The inequality constraint \(\alpha \leq C^2(\rho(t)) \leq \beta\) introduces state constraints that enforce the preservation of quantum coherence during the evolution of the system, where \({C}^2(\rho)\) represents a coherence-related measure of the quantum state \(\rho\) defined in \eqref{eq:C}.

In this setting, all sets are assumed to be Lebesgue measurable, and all functions are assumed to be Lebesgue measurable and integrable. The objective is to determine an optimal pair \((\rho^\star, u^\star)\) that minimizes the cost functional \(J\) while satisfying the dynamic equations, 
and state constraints.

Before presenting the optimality conditions we first analyze the feasibility of the optimal control problem (P), that is, if there is a solution $u^\star(t)$ that solves (P). The following result holds.

\begin{prop}
Let the initial condition $\rho_0$ be such that $C^2(\rho_0) \in [\alpha, \beta]$. Consider also that the time-evolution of 
\begin{equation}\label{Phi}
  \Phi=-2i\sum_{j < k} \Big(\langle \delta_{jk}^{\text{Re}} \rangle_\rho \langle[\delta_{jk}^{\text{Re}}, H_c] \rangle_\rho +\langle \delta_{jk}^{\text{Im}} \rangle_\rho \langle[\delta_{jk}^{\text{Im}}, H_c] \rangle_\rho \Big) 
\end{equation}
is bounded away from zero for every $t\in [t_0, t_f]$ and $\Omega$ is sufficiently large. Then, there exists a nonempty set of admissible input signals $\mathcal{\bar U}([t_0, t_f])$ such that 
$$
\mathcal{\bar U}([t_0, t_f])= \big\{u:[t_0, t_f]\to \mathcal{U}: 
C^2(\rho(t)) \in [\alpha, \beta], \forall t\in [t_0, t_f] \big\}
$$ 
and in particular an optimal $u^\star[t_0, t_f])\in \mathcal{\bar U}([t_0, t_f])$.
\end{prop}
\begin{proof}
The proof follows by showing that it is always possible to pick an input signal $u(t)$ that makes $C^2(\rho)$ increase/decrease in time when the coherence is at the lower/upper bound. To show this, we start to compute the evolution of $C^2(\rho)$, which gives
\begin{equation*}
\begin{aligned}
\Tr&(\frac{\partial C^2(\rho)}{\partial t}) =\Tr (\frac{\partial C^2(\rho)}{\partial \rho} \dot{\rho}) \\
&=2 \Tr \big (
(\sum_{j < k} \delta_{jk}^{\text{Re}}\langle \delta_{jk}^{\text{Re}} \rangle_\rho  + 
\delta_{jk}^{\text{Im} }\langle \delta_{jk}^{\text{Im}} \rangle_\rho ) (-i[H,\rho]+\mathcal{L}(\rho)) \big ) \\
&=2\sum_{j < k} \langle \delta_{jk}^{\text{Re}} \rangle_\rho  \Big(\Tr(-i[\delta_{jk}^{\text{Re}}, H] \rho) 
+\Tr (\delta_{jk}^{\text{Re}} \mathcal{L}(\rho))\Big) \\
&\quad+\langle \delta_{jk}^{\text{Im}} \rangle_\rho  (\Tr(-i[\delta_{jk}^{\text{Im}}, H] \rho)
 +\Tr (\delta_{jk}^{\text{Im}} \mathcal{L}(\rho))\\
&=2\sum_{j < k} \langle \delta_{jk}^{\text{Re}} \rangle_\rho  \Big(\Tr(-i[\delta_{jk}^{\text{Re}}, H_0] \rho) + \Tr(-iu(t)[\delta_{jk}^{\text{Re}}, H_c] \rho)\\
&\quad+\Tr (\delta_{jk}^{\text{Re}} \mathcal{L}(\rho))\Big)+\langle \delta_{jk}^{\text{Im}} \rangle_\rho  \big(\Tr(-i[\delta_{jk}^{\text{Im}}, H_0] \rho) \\
&\quad+ \Tr(-iu(t)[\delta_{jk}^{\text{Im}}, H_c] \rho)
+\Tr (\delta_{jk}^{\text{Im}} \mathcal{L}(\rho))\big)
\end{aligned}
\end{equation*}
From which one concludes that
$\Tr(\frac{\partial C^2(\rho)}{\partial t}) = \Phi u + \Delta$, 
where 
\begin{equation*}
\begin{aligned}
 \Delta&=2\sum_{j < k} \langle \delta_{jk}^{\text{Re}} \rangle_\rho  \Tr(-i[\delta_{jk}^{\text{Re}}, H_0] \rho+\delta_{jk}^{\text{Re}} \mathcal{L}(\rho))\\
 &\quad +\langle \delta_{jk}^{\text{Im}} \rangle_\rho \Tr(-i[\delta_{jk}^{\text{Im}}, H_0] \rho+ \delta_{jk}^{\text{Im}} \mathcal{L}(\rho))   
\end{aligned}
\end{equation*}
Thus, the result follows since $\Phi$ and $\Delta$ are bounded and $\Phi$ is bounded away from zero by assumption.
\end{proof}

\begin{remark}
The assumption of $|\Phi|$ being bounded away from zero depends on control Hamiltonian $H_C$. In particular, 
from \eqref{Phi} one can immediately conclude that at least one of the coherence operators must not commute with $H_c$.
\end{remark}

\section{Pontryagin’s Minimum Principle in Gamkrelidze Form}
Let the pair $\left(\rho,u \right)$ be a feasible process, that is, $u(t)$ and $\rho(t)$ satisfy for a.e. $t\in[t_0, t_f]$ the Lindblad dynamics with the stated initial condition, the state and the input constraints. The pair $\big(\rho^{\star},u^{\star} \big)$ is termed as optimal and therefore solves (P) if the value of the cost function $J$ is the minimum possible over the set of all feasible processes. In order to deal with the indicated problem, we indicate two Lagrangian multipliers: $\pi$ and $\mu$, where $\pi:\left[ t_0,t_f \right]\to {{\mathbb{C}}^{N\times N}}$ is the time-varying Lagrange multiplier matrix, whose elements are called the costates of the system, and $\mu$ is a time varying scalar Lagrange multiplier enforcing the coherence constraint. 
To satisfy the coherence state constraints on $C(\rho)$, we define $\Tilde C{(\rho)}$ as the following
\begin{equation}\label{tildec}
\Tilde C{(\rho)}=\left( \begin{matrix}
   {C}^2\left( \rho  \right)-\beta  \\
   \alpha -{C}^2\left( \rho  \right)  \\
\end{matrix} \right)
\end{equation}
such that $\left\{\Tilde C{(\rho)}  \le 0, \forall t\in \left[  t_0,t_f\right] \right\}$. 
We will now derive the necessary optimality conditions.

\begin{prop}
Consider the optimal control problem (P). Let $u^\star(t)$ be an optimal control and $\rho^\star(t)$ be the corresponding state trajectory. Then, there exists a multiplier $\pi^\star(t)$ that, together with $\mu^\star(t)$, satisfies the necessary conditions according to Gamkrelidze form of PMP. Based on the given setup, we express the extended Pontryagin Hamiltonian function as follows:
\begin{equation}\label{PH}
    \begin{aligned}
        \mathcal{H}(\rho,\pi,u, \mu ) = \Tr &\left(\left(\pi + \mu \sum_{j < k} \left( \langle \delta_{jk}^{\text{Re}} \rangle _\rho \delta_{jk}^{\text{Re}} + \langle \delta_{jk}^{\text{Im}} \rangle_\rho \delta_{jk}^{\text{Im}} \right)\right)^\dagger \right. \\
        &\left. 
        \left (
        -\frac{i}{\hbar} [H(t), \rho] + \mathcal{L}(\rho) \right) \right) + u^{2}(t)
    \end{aligned}
\end{equation}
Specifically, we have the condition
\begin{equation*}
    \mathcal{H}(\rho^\star,\pi^\star,u^\star, \mu^\star, t) \le \mathcal{H}(\rho^\star,\pi^\star,u, \mu^\star, t),
\end{equation*}
for all $t \in [t_0,t_f]$, and for all feasible controls $u \in \Omega$. Moreover, the control $u(t)$ must satisfy
\begin{equation}
\begin{aligned}
    \frac{\partial \mathcal{H}}{\partial u} = \frac{-i}{\hbar} \Tr &\Bigg( \Big( \pi^\star + \mu \sum_{j<k} \left( \langle \delta_{jk}^{\text{Re}} \rangle_{\rho^\star}  \delta_{jk}^{\text{Re}} + \langle \delta_{jk}^{\text{Im}} \rangle_{\rho^\star}  \delta_{jk}^{\text{Im}} \right) \Big)^\dagger \\
    & [H_C, \rho^\star] \Bigg) + 2 u(t)=0
\end{aligned}    
\end{equation}
The adjoint equation is derived as follows:
\begin{equation}\label{adj}
\begin{aligned}
\dot{\pi}^\star &= -\frac{\partial \mathcal{H}}{\partial \rho} = \chi\big( \pi^\dagger \big) + \mu \sum\limits_{j<k} \big( \delta_{jk}^{\text{Re}} \langle \chi({\delta_{jk}^{\text{Re}}}^\dagger) \rangle_{\rho^\star} \\
   &+ \delta_{jk}^{\text{Im}} \langle \chi({\delta_{jk}^{\text{Im}}}^\dagger) \rangle_{\rho^\star} 
    + \langle \delta_{jk}^{\text{Re}} \rangle_{\rho^\star} \chi({\delta_{jk}^{\text{Re}}}^\dagger) 
   + \langle \delta_{jk}^{\text{Im}} \rangle_{\rho^\star} \chi({\delta_{jk}^{\text{Im}}}^\dagger) \big),
\end{aligned}
\end{equation}
and its final condition implies that $\pi (t_f)=0_{N \times N}$.
The function $\chi$ for a generic variable $\phi$ is defined as
\begin{equation}\label{chi}
\chi(\phi) = \frac{-i}{\hbar}[\phi, H(t)] + \mathcal{L}(\phi)+\mathcal{L}_0(\phi),
\end{equation}
with $\mathcal{L}_0(\phi)=
\sum\limits_{k} \gamma_k (L_k ^{\dagger } \phi L_{k}- L_k \phi L_{k}^{\dagger})$.
\end{prop}
\begin{proof}
We follow the necessary optimality conditions of the PMP in the form of Gamkrelidze for state-constrained optimal control problems, as studied in \cite{arutyunov2011maximum}. We extend these arguments to the quantum context for matrix-valued dynamics. 
To this end, the extended Pontryagin Hamiltonian is defined as:
\begin{equation*}
\mathcal{H}(\rho, \pi, u, \mu, t) = \Tr\Big(\pi^\dagger \dot{\rho} + (\bm{\mu}^T \frac{\partial \tilde{C}(\rho)}{\partial \rho})^\dagger \dot{\rho}\Big) + u^2(t),
\end{equation*}
where $
\bm{\mu}^T \frac{\partial \tilde{C}(\rho)}{\partial \rho} = \left( \mu_1 \quad \mu_2 \right) \begin{pmatrix} \frac{\partial C^2(\rho)}{\partial \rho} \\ - \frac{\partial C^2(\rho)}{\partial \rho} \end{pmatrix}
$, simplifying to
$$
(\mu_1 - \mu_2) \frac{\partial C^2(\rho)}{\partial \rho} = 2(\mu_1 - \mu_2) \sum_{j < k} \left( \langle \delta_{jk}^{\text{Re}} \rangle_\rho \delta_{jk}^{\text{Re}} + \langle \delta_{jk}^{\text{Im}} \rangle_\rho \delta_{jk}^{\text{Im}} \right)
$$ From this, we define $\mu = 2(\mu_1 - \mu_2)$. Thus, the Pontryagin Hamiltonian is formulated as shown in equation \eqref{PH}, from which the necessary conditions according to the PMP are derived. In the following, we show the derivation of the adjoint equation \eqref{adj}. Due to the complexity of computation, we break $\mathcal{H}$ into four main terms, such that
$\mathcal{H}=\mathcal{H}_1+\mathcal{H}_2+\mathcal{H}_3+\mathcal{H}_4$, where
\begin{equation*}
\mathcal{H}_1=-\frac{i}{\hbar}\Tr(\pi^\dagger([H,\rho]))=-\frac{i}{\hbar}\big(\Tr(\pi^\dagger H\rho)-\Tr(\pi^\dagger \rho H)\big)
\end{equation*}
\begin{equation*}
\begin{aligned}
\mathcal{H}_2&=\Tr(\pi^\dagger \mathcal{L}(\rho))=\sum\limits_{k} \gamma_k \big(\Tr(\pi^\dagger L_k \rho L_k^\dagger)-\frac{1}{2} \Tr(\pi^\dagger L_k^\dagger L_k \rho)\\
&\quad-\frac{1}{2} \Tr(\pi^\dagger \rho L_k^\dagger L_k)\big)
\end{aligned}
\end{equation*} 

\begin{equation*}
\begin{aligned}
&\mathcal{H}_3=  \mu  \Tr( ( \sum_{j<k}  \langle \delta_{jk}^{\text{Re}} \rangle_{\rho}  {\delta_{jk}^{\text{Re}}} ^\dagger + \langle \delta_{jk}^{\text{Im}} \rangle_{\rho}  {\delta_{jk}^{\text{Im}}}^\dagger ) (-\frac{i}{\hbar}[H,\rho]))\\
&=\frac{-i \mu}{\hbar} \sum_{j<k} \langle \delta_{jk}^{\text{Re}} \rangle_{\rho} (\Tr({\delta_{jk}^{\text{Re}}} ^\dagger H \rho)- \Tr({\delta_{jk}^{\text{Re}}} ^\dagger \rho H)) + \\
&\langle \delta_{jk}^{\text{Im}} \rangle_{\rho} (\Tr({\delta_{jk}^{\text{Im}}} ^\dagger H \rho)- \Tr({\delta_{jk}^{\text{Im}}} ^\dagger \rho H))
\end{aligned}
\end{equation*}
\begin{equation*}
\begin{aligned}
 &\mathcal{H}_4= \mu \Tr (\sum_{j<k} \left( \langle \delta_{jk}^{\text{Re}} \rangle _\rho {\delta_{jk}^{\text{Re}}}^\dagger + \langle \delta_{jk}^{\text{Im}} \rangle_\rho {\delta_{jk}^{\text{Im}}}^\dagger \right)\mathcal{L}(\rho))= \\
&\mu  \sum_{j<k}  \langle \delta_{jk}^{\text{Re}} \rangle _\rho  
\sum_k \gamma_k  \left (\Tr({\delta_{jk}^{\text{Re}}}^\dagger L_k \rho L_k^\dagger)
- \frac{1}{2} \Tr({\delta_{jk}^{\text{Re}}}^\dagger L_k^\dagger L_k \rho) \right. \\
&  \left. - \frac{1}{2} \Tr( {\delta_{jk}^{\text{Re}}}^\dagger \rho L_k^\dagger L_k )\right) + \mu  \sum_{j<k}  \langle \delta_{jk}^{\text{Im}} \rangle _\rho   \sum_k \gamma_k  \left (\Tr({\delta_{jk}^{\text{Im}}}^\dagger L_k \rho L_k^\dagger)\right.\\
&- \frac{1}{2} \Tr({\delta_{jk}^{\text{Im}}}^\dagger L_k^\dagger L_k \rho)    \left. - \frac{1}{2} \Tr( {\delta_{jk}^{\text{Im}}}^\dagger \rho L_k^\dagger L_k )\right)
\end{aligned}
\end{equation*}
Now, we compute $\frac{\partial \mathcal{H}}{\partial \rho}$ by computing the derivative of all above four terms.
\begin{equation*}
    \frac{\partial \mathcal{H}_1}{\partial \rho}= -\frac{i}{\hbar}(\pi ^\dagger H -H \pi ^\dagger)=-\frac{i}{\hbar}[\pi^ \dagger , H]
\end{equation*}
\begin{equation*}
    \begin{aligned}
        & \frac{\partial \mathcal{H}_2}{\partial \rho}=\sum\limits_{k} \gamma_k (L_k^\dagger \pi^\dagger L_k-\frac{1}{2} \pi^\dagger L_k^\dagger L_k -\frac{1}{2} L_k^\dagger L_k \pi^\dagger +L_k\pi^\dagger L_k^\dagger\\
&-L_k\pi^\dagger L_k^\dagger)=\mathcal{L}(\pi^\dagger)+\mathcal{L}_0(\pi^\dagger)
    \end{aligned}
\end{equation*}
{\small \begin{equation*}
    \begin{aligned}
        &\frac{\partial \mathcal{H}_3}{\partial \rho}=\frac{-i \mu}{\hbar} \sum_{j<k} \delta_{jk}^{\text{Re}} (\Tr({\delta_{jk}^{\text{Re}}} ^\dagger H \rho)- \Tr({\delta_{jk}^{\text{Re}}} ^\dagger \rho H)) + \\
&\delta_{jk}^{\text{Im}} (\Tr({\delta_{jk}^{\text{Im}}} ^\dagger H \rho)- \Tr({\delta_{jk}^{\text{Im}}} ^\dagger \rho H))+ \langle \delta_{jk}^{\text{Re}} \rangle_{\rho} ({\delta_{jk}^{\text{Re}}} ^\dagger H - H{\delta_{jk}^{\text{Re}}} ^\dagger) \\
&+\langle \delta_{jk}^{\text{Im}} \rangle_{\rho} ({\delta_{jk}^{\text{Im}}} ^\dagger H- H{\delta_{jk}^{\text{Im}}} ^\dagger)= \frac{-i \mu}{\hbar} \sum_{j<k} \delta_{jk}^{\text{Re}} \Tr([{\delta_{jk}^{\text{Re}}} ^\dagger, H] \rho)\\
&+\delta_{jk}^{\text{Im}} \Tr([{\delta_{jk}^{\text{Im}}} ^\dagger, H] \rho)
+\langle \delta_{jk}^{\text{Re}} \rangle_{\rho} [{\delta_{jk}^{\text{Re}}} ^\dagger, H ]
+\langle \delta_{jk}^{\text{Im}} \rangle_{\rho} [{\delta_{jk}^{\text{Im}}} ^\dagger, H ]
 \end{aligned}
\end{equation*}}
\small{\begin{equation*}
\begin{aligned}
\frac{\partial \mathcal{H}_4}{\partial \rho}&= 
\mu  \sum_{j<k}  \delta_{jk}^{\text{Re}}    
\sum_k \gamma_k  \big (\Tr({\delta_{jk}^{\text{Re}}}^\dagger L_k \rho L_k^\dagger)
- \frac{1}{2} \Tr({\delta_{jk}^{\text{Re}}}^\dagger L_k^\dagger L_k \rho)  \\
&   - \frac{1}{2} \Tr( {\delta_{jk}^{\text{Re}}}^\dagger \rho L_k^\dagger L_k )\big) + \mu  \sum_{j<k}   \delta_{jk}^{\text{Im}}   \sum_k \gamma_k  \big (\Tr({\delta_{jk}^{\text{Im}}}^\dagger L_k \rho L_k^\dagger)\\
&- \frac{1}{2} \Tr({\delta_{jk}^{\text{Im}}}^\dagger L_k^\dagger L_k \rho)   - \frac{1}{2} \Tr( {\delta_{jk}^{\text{Im}}}^\dagger \rho L_k^\dagger L_k )\big) 
+\mu  \sum_{j<k}  \langle \delta_{jk}^{\text{Re}} \rangle _\rho   \\  
&\sum_k \gamma_k  \big (\Tr(L_k^\dagger{\delta_{jk}^{\text{Re}}}^\dagger L_k  )
- \frac{1}{2} \Tr({\delta_{jk}^{\text{Re}}}^\dagger L_k^\dagger L_k)  - \frac{1}{2} \Tr(L_k^\dagger L_k {\delta_{jk}^{\text{Re}}}^\dagger)\big)\\
&  +\mu  \sum_{j<k}  \langle \delta_{jk}^{\text{Im}} \rangle _\rho   \sum_k \gamma_k  \big (\Tr(L_k^\dagger{\delta_{jk}^{\text{Im}}}^\dagger L_k )- \frac{1}{2} \Tr({\delta_{jk}^{\text{Im}}}^\dagger L_k^\dagger L_k )  \\
&- \frac{1}{2} \Tr(L_k^\dagger  L_k {\delta_{jk}^{\text{Im}}}^\dagger \rho  )\big) \\
&=\mu \big( \sum_{j<k}  \delta_{jk}^{\text{Re}} \big(\Tr(\mathcal{L}({\delta_{jk}^{\text{Re}}}^\dagger)\rho ) +\big( 
 \Tr(\mathcal{L}_0({\delta_{jk}^{\text{Re}}}^\dagger) \rho )\big)\\
 &\quad+\sum_{j<k}  \delta_{jk}^{\text{Im}} \big(\Tr(\mathcal{L}({\delta_{jk}^{\text{Im}}}^\dagger  )\rho) +\Tr(\mathcal{L}_0({\delta_{jk}^{\text{Im}}}^\dagger)\rho) \big)\\
&\quad+ \sum_{j<k} \langle \delta_{jk}^{\text{Re}} \rangle _\rho  \big( \mathcal{L}({\delta_{jk}^{\text{Re}}}^\dagger) +\mathcal{L}_0({\delta_{jk}^{\text{Re}}}^\dagger)\big) \\
&\quad+\sum_{j<k} \langle \delta_{jk}^{\text{Im}} \rangle _\rho  \big( \mathcal{L}({\delta_{jk}^{\text{Im}}}^\dagger) +  
\mathcal{L}_0({\delta_{jk}^{\text{Im}}}^\dagger)\big)\big)
\end{aligned}
\end{equation*}}
By summing the four derivatives, we compute the adjoint equation $-\dot{\pi}=\frac{\partial \mathcal{H}}{\partial \rho}=\frac{\partial \mathcal{H}_1}{\partial \rho}+\frac{\partial \mathcal{H}_2}{\partial \rho}+\frac{\partial \mathcal{H}_3}{\partial \rho}+\frac{\partial \mathcal{H}_4}{\partial \rho}$
\small{\begin{equation*}
    \begin{aligned}
        -&\dot{\pi}= -\frac{i}{\hbar}[\pi^ \dagger , H] +\mathcal{L}(\pi^\dagger)+\mathcal{L}_0(\pi^\dagger) \\
        &-\frac{i \mu}{\hbar} \sum_{j<k} \delta_{jk}^{\text{Re}} \Tr([{\delta_{jk}^{\text{Re}}} ^\dagger, H] \rho) 
       \\
       &+\mu  \sum_{j<k}  \delta_{jk}^{\text{Re}} \left(\Tr(\mathcal{L}({\delta_{jk}^{\text{Re}}}^\dagger)\rho)  \right.
+\mu  \sum_{j<k}  \delta_{jk}^{\text{Re}} \Tr(\mathcal{L}_0({\delta_{jk}^{\text{Re}}}^\dagger)\rho ) \\
&-\frac{i \mu}{\hbar} \sum_{j<k} \langle \delta_{jk}^{\text{Re}} \rangle_{\rho} [{\delta_{jk}^{\text{Re}}} ^\dagger, H ]
+\mu \sum_{j<k} \langle \delta_{jk}^{\text{Re}} \rangle _\rho  \left( \mathcal{L}({\delta_{jk}^{\text{Re}}}^\dagger) +\mathcal{L}_0({\delta_{jk}^{\text{Re}}}^\dagger)\right) \\
       &-\frac{i \mu}{\hbar} \sum_{j<k} \delta_{jk}^{\text{Im}} \Tr([{\delta_{jk}^{\text{Im}}} ^\dagger, H] \rho) + \mu  \sum_{j<k}  \delta_{jk}^{\text{Re}} \left(\Tr(\mathcal{L}({\delta_{jk}^{\text{Im}}}^\dagger)\rho)  \right.
       \\
       &
+\mu  \sum_{j<k}  \delta_{jk}^{\text{Im}} \Tr(\mathcal{L}_0({\delta_{jk}^{\text{Re}}}^\dagger)\rho)  -
\frac{i \mu}{\hbar} \sum_{j<k} \langle \delta_{jk}^{\text{Im}} \rangle_{\rho} [{\delta_{jk}^{\text{Im}}} ^\dagger, H ]\\
&
+\mu \sum_{j<k} \langle \delta_{jk}^{\text{Im}} \rangle _\rho  \left( \mathcal{L}({\delta_{jk}^{\text{Re}}}^\dagger) +\mathcal{L}_0({\delta_{jk}^{\text{Re}}}^\dagger)\right)
\end{aligned}
\end{equation*}}
which simplifies to
\begin{equation*}
    \begin{aligned}
        -\dot{\pi} = &-\frac{i}{\hbar}[\pi^ \dagger , H] +\mathcal{L}(\pi^\dagger)+\mathcal{L}_0(\pi^\dagger)\\
       &+\mu \sum_{j<k} \left ( \delta_{jk}^{\text{Re}} 
       \langle   -\frac{i}{\hbar}[{\delta_{jk}^{\text{Re}}} ^\dagger, H] \right.+\mathcal{L}({\delta_{jk}^{\text{Re}}}^\dagger) +\mathcal{L}_0({\delta_{jk}^{\text{Re}}}^\dagger)\rangle_\rho \\&+\delta_{jk}^{\text{Im}} 
       \langle-\frac{i}{\hbar}[{\delta_{jk}^{\text{Im}}} ^\dagger, H]+\mathcal{L}({\delta_{jk}^{\text{Im}}}^\dagger) 
    +\mathcal{L}_0({\delta_{jk}^{\text{Im}}}^\dagger)\rangle_\rho\\
       &+\langle \delta_{jk}^{\text{Re}} \rangle_{\rho} \left( -\frac{i}{\hbar}  [{\delta_{jk}^{\text{Re}}} ^\dagger, H ]+\mathcal{L}({\delta_{jk}^{\text{Re}}}^\dagger) +\mathcal{L}_0({\delta_{jk}^{\text{Re}}}^\dagger)\right)\\
       &+\left. \langle \delta_{jk}^{\text{Im}} \rangle_{\rho}  \left( -\frac{i}{\hbar}  [{\delta_{jk}^{\text{Im}}} ^\dagger, H ]+\mathcal{L}({\delta_{jk}^{\text{Im}}}^\dagger) +\mathcal{L}_0({\delta_{jk}^{\text{Im}}}^\dagger) \right)  \right)      
    \end{aligned}
\end{equation*} 
By defining the function \eqref{chi}, the adjoint equation \eqref{adj} is proved. In addition, the transversality condition, which applies at the final time $t_f$, imposes that
\begin{equation*}
\pi(t_f) = -\frac{\partial J}{\partial \rho_f} = 0_{N \times N}.
\end{equation*}
\end{proof}
\begin{prop}
Let the quadruple $(\rho^\star, \pi^\star, u^\star,\mu^\star)$ be an optimal trajectory of problem (P) for the case where the state constraints are not active. Then,  this quadruple is a local minimizer of the extended Pontryagin Hamiltonian \eqref{PH}.
\end{prop}
\begin{proof}
This follows from the fact that the second derivative of the extended Pontryagin Hamiltonian with respect to the control $u$ is $\frac{\partial^2 \mathcal{H}}{\partial u^2} =2>0$.
\end{proof}

\begin{prop}
Consider problem (P) where the state constraints are active for some period of time $\tau \in [t_0,t_f]$ and let $\bar\rho(t)$, $t\in\tau$ be the corresponding density matrix for each $C^2(\bar\rho)=\alpha$ (or $C^2(\bar\rho)=\beta$ for the upper bound case). Then, as long as 
the control Hamiltonian $H_c$ does not commute with $\bar\rho$,
there exists a feasible boundary control $u(t)$ for $t\in\tau$  given by 
\begin{equation} \label{eq:u_const}
u(t)= \frac{\big(-\frac{i}{\hbar} [H_0, \rho] + \mathcal{L}(\rho)\big)}{\frac{i}{\hbar}[H_c, \rho]}.      
\end{equation}
\end{prop}
\begin{proof}
By imposing $\Gamma=\Tr\big (\bm{\mu}^T \frac{\partial \tilde{C}(\rho)}{\partial \rho}\dot{\rho}\big)$ equal to zero, we compute the control $u$ that keeps the lower bound of state constraint, which is  
\begin{equation*}
\Gamma= 2\mu \sum_{j < k} \big( \langle \delta_{jk}^{\text{Re}} \rangle_\rho \delta_{jk}^{\text{Re}} + \langle \delta_{jk}^{\text{Im}} \rangle_\rho \delta_{jk}^{\text{Im}} \big) \big(-\frac{i}{\hbar} [H(t), \rho] + \mathcal{L}(\rho)\big)
\end{equation*}
From this it follows \eqref{eq:u_const}.
\end{proof}

\begin{algorithm}\label{alg1}
\small
\caption{Solving the State-Constrained Quantum Optimal Control Problem with Coherence Preservation.}
\begin{algorithmic}[1]
    \State \textbf{Step 1 - Initialization}
    \State Set number of time steps $N$ and initialize the iteration counter $j = 0$.
    \For{$m = 0,\dots, N-1$}
        \State Initialize control $u_m^j$.
        \State Initialize multiplier $\mu_m^j$.
    \EndFor
    \State Set initial state $\rho_0$ at time $t_0$.
    \State Set the adjoint variable $\pi_f$ at final time $t_f$ according to the transversality condition. 
    \State Set convergence tolerance $\varepsilon$, and learning rates $\eta_i$, $\eta_d$.
    
    \State \textbf{Step 2 - Computation of the state trajectory}
    \For{$m = 0,\dots,N-1$}
        \State Update the state $\rho_m^j$ by solving the Lindblad master equation numerically using a time-stepping method 
        \[
        \dot{\rho}_m^j = -\frac{i}{\hbar} [H_m^j, \rho_m^j] + \mathcal{L}(\rho_m^j)
        \]
    \EndFor

    \State \textbf{Step 3 - Computation of the adjoint trajectory}
    \For{$m = N-1$ \textbf{to} $0$}
        \State Update the adjoint variable $\pi_m^j$ by solving the adjoint equation numerically using the chosen time-stepping method 
        \[
        \pi_m^j = \chi(\pi^\dagger_m) + \mu_m \sum_{j<k} \left( \delta_{jk}^{Re} \langle \chi(\delta_{jk}^{Re}) \rangle_{\rho_m^j} + \delta_{jk}^{Im} \langle \chi(\delta_{jk}^{Im}) \rangle_{\rho_m^j} \right)
        \]
    \EndFor  

    \State \textbf{Step 4 - Update of Lagrange multipliers for state constraint}
    \For{$m =0, \dots, N-1$}
        \State Compute the coherence function \( C^2(\rho_m^j) \)
        \If{ \( C^2(\rho_m^j) < \beta \) }
            \State Set \( \mu_{1,m}^{j} = \mu_{1,m-1}^{j} \)
        \ElsIf{ \( C^2(\rho_m^j) > \alpha \) }
            \State Set \( \mu_{2,m}^{j} = \mu_{2,m-1}^{j} \)
        \ElsIf{ \( C^2(\rho_m^j) = \beta \) }
              \State Update \( \mu_{1,m}^{j} = \mu_{1,m-1}^{j} + \eta_i
            ({ C^2(\rho_m^j)-C^2(\rho_{m-1}^j)})
             \)
        \ElsIf{ \( C^2(\rho_m^j) = \alpha \) }
             \State Update \( \mu_{2,m}^{j} = \mu_{2,m-1}^{j} - \eta_d
            ({ C^2(\rho_m^j)-C^2(\rho_{m-1}^j)})
             \)
        \EndIf
        \State Update the combined Lagrange multiplier:
        \[
        \mu_{m}^{j} = 2 \left( \mu_{1,m}^{j} - \mu_{2,m}^{j} \right)
        \]        
    \EndFor
    \State \textbf{Step 5 - Computation of the Pontryagin Hamilton function}
    \For{$m = 0, \dots, N-1$}
        \State Compute the Pontryagin Hamiltonian function \(\mathcal{H}_m^j\) as:
        \begin{equation*}
            \begin{aligned}
               &\mathcal{H}_m^j = \Tr\left( \left( \pi_m^j + \mu_m^j \sum_{j<k} \left( \langle \delta_{jk}^{\text{Re}}, \rho_m^j \rangle \delta_{jk}^{\text{Re}} + \langle \delta_{jk}^{\text{Im}}, \rho_m^j \rangle \delta_{jk}^{\text{Im}} \right) \right)^\dagger \right. \\          
               &\left. \left( -\frac{i}{\hbar} [H_m^j, \rho_m^j] + \mathcal{L}(\rho_m^j) \right) \right) + {u_m^j}^2  \\
            \end{aligned}
        \end{equation*}
    \EndFor
        \State \textbf{Step 6 - Update the control function}
    \For{$m = 0, \dots, N-1$}
     \State 
     Update the control values \(u_m^{j+1}\) using convex combination, where $\zeta_1+\zeta_2=1$
            \begin{equation*}
            \begin{aligned}
             u_m^{j+1} &= \zeta_1 \frac{i}{2\hbar} \Tr \Big( \big( \pi_m^j - \mu_m^j \sum_{j<k} \big( \langle \delta_{jk}^{\text{Re}}\rangle _{\rho_m^j} \delta_{jk}^{\text{Re}} + \langle \delta_{jk}^{\text{Im}}\rangle _{\rho_m^j }\delta_{jk}^{\text{Im}} \big) \big)^\dagger \\
            &\quad \left[{H_C}_m^j, \rho_m^j \right] \Big)+\zeta_2 u_m^j
            \end{aligned}
        \end{equation*}
        \EndFor
    \algstore{myalg}
\end{algorithmic}
\end{algorithm}

\begin{algorithm}
\begin{algorithmic}[1]
\small
\algrestore{myalg}
\State \textbf{Step 7 - Stopping Test}
    \State Define tolerance values: \(\varepsilon_u\) and \(\varepsilon_\mu\) (small positive numbers).
    \For{$m = 0, \dots, N-1$}

         \State Check the convergence criteria for the state trajectory:
        \[ |\rho_m^{j+1} - \rho_m^j| < \varepsilon_\rho \]
        \State Check the convergence criteria for the adjoint variable:
        \[ |\pi_m^{j+1} - \pi_m^j| < \varepsilon_\pi \]
        \State Check the convergence criteria for the control:
        \[ |u_m^{j+1} - u_m^j| < \varepsilon_u \]
        \State Check the convergence criteria for the Lagrange multipliers:
        \[ |\mu_{m}^{j+1} - \mu_m^j| < \varepsilon_\mu \]
    \EndFor
    \State
   \textbf{if}
    $\max\{|\rho_m^{j+1} - \rho_m^j|, |\pi_m^{j+1} - \pi_m^j|, |u_m^{j+1} - u_m^j|, |\mu_{m}^{j+1} - \mu_m^j|\}<\epsilon$ 
    \State \textbf{Stop}
    \State \textbf{else if} increment \(j = j + 1\) and go to \textbf{Step 2}.
\end{algorithmic}
\end{algorithm}
Algorithm 1 provides a step-by-step method for solving the state-constrained quantum optimal control problem with coherence preservation. It initializes control variables and multipliers, computes state and adjoint trajectories through numerical methods, updates Lagrange multipliers to enforce state constraints, calculates the Pontryagin Hamiltonian, and iteratively updates control functions. The algorithm concludes with a convergence check to ensure optimal control, state trajectory, and multiplier conditions are met before termination.

\begin{remark} [Control Constraints] 
In addition to the state constraints handled by the presented algorithm, control constraints can also be efficiently addressed by applying the method of saturation functions, as described in \cite{graichen2010handling}. This method transforms the control constraints into smooth functions that approximate the behavior of the original constraints, ensuring that the control variables remain within their prescribed limits without imposing hard bounds. The approach involves dynamically extending the system by incorporating these saturation functions into the Pontryagin Hamiltonian, thus reformulating the original problem as an unconstrained optimization problem. This strategy can be seamlessly integrated into the current algorithm, ensuring that both state and control constraints are respected during the optimization process, while maintaining numerical stability and preventing constraint violations.
\end{remark}

\section{Application of the Algorithm to a Qutrit System}
We now apply the developed algorithm to the specific case of a qutrit, a three-level quantum system, subjected to Markovian decoherence. The qutrit consists of three quantum states, \( |0\rangle \), \( |1\rangle \), and \( |2\rangle \), and is described by a density matrix \( \rho \in \mathbb{C}^{3\times 3} \): $\rho = \begin{pmatrix}
\rho_{00} & \rho_{01} & \rho_{02} \\
\rho_{10} & \rho_{11} & \rho_{12} \\
\rho_{20} & \rho_{21} & \rho_{22}
\end{pmatrix}$.
The basis states for the qutrit system are represented as follows:
\[
|0\rangle = \begin{pmatrix} 1 \\ 0 \\ 0 \end{pmatrix}, \quad |1\rangle = \begin{pmatrix} 0 \\ 1 \\ 0 \end{pmatrix}, \quad |2\rangle = \begin{pmatrix} 0 \\ 0 \\ 1 \end{pmatrix}.
\]
The off-diagonal elements \( \rho_{01} \), \( \rho_{02} \), and \( \rho_{12} \) represent the coherence between the different pairs of states. In this qutrit system, the coherence between states \( |0\rangle \) and \( |1\rangle \) can be captured by the following operators:
\begin{align*}
\delta_{01}^{\text{Re}} &= |0\rangle \langle 1| + |1\rangle \langle 0|, \\
\delta_{01}^{\text{Im}} &= -i \left( |0\rangle \langle 1| - |1\rangle \langle 0| \right).
\end{align*}
Similarly, the coherence between states \( |0\rangle \) and \( |2\rangle \), and between \( |1\rangle \) and \( |2\rangle \), is described by their respective real and imaginary coherence operators \( \delta_{02}^{\text{Re}} \), \( \delta_{02}^{\text{Im}} \), and \( \delta_{12}^{\text{Re}} \), \( \delta_{12}^{\text{Im}} \).
The total coherence between states \( |0\rangle \) and \( |1\rangle \) is then computed using the coherence function
\[
C(\rho) = \sqrt{ \left( \langle \delta_{01}^{\text{Re}} \rangle_\rho \right)^2 + \left( \langle \delta_{01}^{\text{Im}} \rangle_\rho \right)^2 }.
\]
The loss of coherence between the two states can be due to decoherence, often caused by interaction with the environment. However, by applying a control field, it is possible to drive the transitions between these states and counteract this decay process.

In the context of this qutrit, we assume the system has energy levels associated with the states \( |0\rangle \), \( |1\rangle \), and \( |2\rangle \), with corresponding energies \( E_0 \), \( E_1 \), and \( E_2 \). The drift Hamiltonian describes the free evolution of the system without external control fields and can be written as follows:
\begin{equation*}
    \begin{aligned}
      H_0 &= \frac{E_2 - E_0}{3}  (|2\rangle \langle 2| - |0\rangle \langle 0|) + \frac{E_2 - E_1}{3} (|2\rangle \langle 2| - |1\rangle \langle 1|)\\
      &\quad + \frac{E_0 - E_1}{3} (|0\rangle \langle 0| - |1\rangle \langle 1|)  
    \end{aligned},
\end{equation*}
where each term represent the population differences between the corresponding energy levels in the system. Thus, the drift Hamiltonian \( H_0 \) captures the natural dynamics of the system, reflecting the energy differences between the three states. We assume that the transition dipole moments for the linearly polarized field are real. Hence, the control field, in the dipole approximation, can be described by a time-dependent envelope of the field $u(t)$, driving frequency $\omega_d$, and initial phase $\phi_d$, hence, the corresponding control Hamiltonian is given by
\[
H_C = u(t) \left( e^{i \phi_d} |0\rangle \langle 1| + e^{-i \phi_d} |1\rangle \langle 0| \right) \cos(\omega_d t)
\]
This control $u(t)$ is designed to influence the system’s coherence and dynamics, allowing precise control over the quantum states.

In the context of our qutrit system, the Lindblad operators describe the decay processes from the excited state \( |2\rangle \) to the two lower states \( |0\rangle \) and \( |1\rangle \), with respective decay rates \( \gamma_0 \) and \( \gamma_1 \). These transitions are modeled by the Lindblad operator, which introduces the dissipative effects into the system's dynamics.

The Lindblad operators \( L_k \) for the qutrit system are defined as
\[
L_0 = |0\rangle \langle 2|, \quad L_1 = |1\rangle \langle 2|, \quad L_{3}=\ket{0}\bra{0}-\ket{1}\bra{1}.
\]
The operator \( L_0 \) represents the transition from the excited state \( |2\rangle \) to the ground state \( |0\rangle \), while \( L_1 \) represents the transition from \( |2\rangle \) to \( |1\rangle \). Additionally, dephasing between the states \( |0\rangle \) and \( |1\rangle \) is introduced via a dephasing operator \( L_{3} \).
Thus, the full Lindblad term that governs the system's evolution under decoherence is given by
\begin{equation} \label{lrho}
\begin{aligned} 
&\mathcal{L}(\rho) = \frac{1}{2} \gamma_0 \left( 2 |0\rangle \langle 2| \rho |2\rangle \langle 0| - |2\rangle \langle 2| \rho - \rho |2\rangle \langle 2| \right)+ \frac{1}{2} \gamma_1\\ & \left( 2 |1\rangle \langle 2| \rho |2\rangle \langle 1| - |2\rangle \langle 2| \rho - \rho |2\rangle \langle 2| \right)+\frac{1}{2} \gamma_d \left( 2 \left( |0\rangle \langle 0| - |1\rangle \langle 1| \right) \right.\\
&\left.\rho \left( |0\rangle \langle 0| - |1\rangle \langle 1| \right) - \left( |0\rangle \langle 0| - |1\rangle \langle 1| \right)^2 \rho - \rho \left( |0\rangle \langle 0| - |1\rangle \langle 1| \right)^2 \right) 
\end{aligned} \end{equation}
This formulation reflects the combined effects of energy dissipation and coherence loss due to interactions with the environment.
\begin{prop}
The indicated Lindbladian \( \mathcal{L}(\rho) \) in \eqref{lrho} is non-unital
\end{prop}
\begin{proof}
To confirm that the Lindbladian \( \mathcal{L} \) is non-unital, we must show that \( \mathcal{L}(I) \), when acting on the identity matrix \( I \) for a qutrit system, yields a non-zero result . 
By applying the Lindblad operators to \( I \) and evaluating this expression, one obtains $\mathcal{L}(I) = \text{diag}(\gamma_0, \gamma_0 + \gamma_1, -2\gamma_0 - \gamma_1)$.
Since \( \mathcal{L}(I) \neq 0 \), this confirms that the Lindbladian is non-unital, reflecting that the system's evolution includes irreversible processes. 
\end{proof}

\subsection{Numerical Simulation}

In this section, we describe the numerical simulation of a qutrit subjected to Markovian decoherence and dephasing. 
The decay rates \( \gamma_0 \) and \( \gamma_1 \) are set to 0.1 and 0.001 for the transitions from state \( |2\rangle \) to states \( |0\rangle \) and \( |1\rangle \), respectively. The dephasing rate between states \( |0\rangle \) and \( |1\rangle \), \( \gamma_d \), is set to 0.005. The energy levels are set to $E_0 = 1$, 
$E_1 = 1.5$, and $E_2 = 2$. 
The system's evolution is computed over a time span of 20 a.u. (arbitrary units), and the density matrix is updated at 1000 time points 
to numerically solve the Lindblad equation. In our numerical simulation, we consider initial state $\rho_0 =\left(  \begin{smallmatrix}
0.21 & 0.195-0.195i& 0 \\
0.195+0.195i & 0.78 & 0 \\
0 & 0 & 0.01
\end{smallmatrix} \right) $. First, we consider the system evolution without application of any external control field. This allows analyzing the effects of decay and dephasing captured by the Lindblad operators under
the natural dynamics of the system.
The coherence between states \( |0\rangle \) and \( |1\rangle \) is computed. The simulation results, as shown in Fig.~\ref{fig:NC}(a), indicate that the coherence decays over time. This behavior aligns with theoretical expectations for a system subject to Markovian decoherence, where coherence is gradually lost due to interactions with the environment. The evolution of the key density matrix elements, the real and imaginary parts of the off-diagonals, $\rho_{01}$ and $\rho_{10}$, is plotted in Fig.~\ref{fig:NC}(b). The real and imaginary components of the coherence between states \( |0\rangle \) and \( |1\rangle \), represented by Re(\( \rho_{01} \)) and Im(\( \rho_{01} \)), show oscillations that gradually dampen, reflecting the loss of coherence due to dephasing. 
\begin{figure}
    \centering
    \includegraphics[scale=0.25]{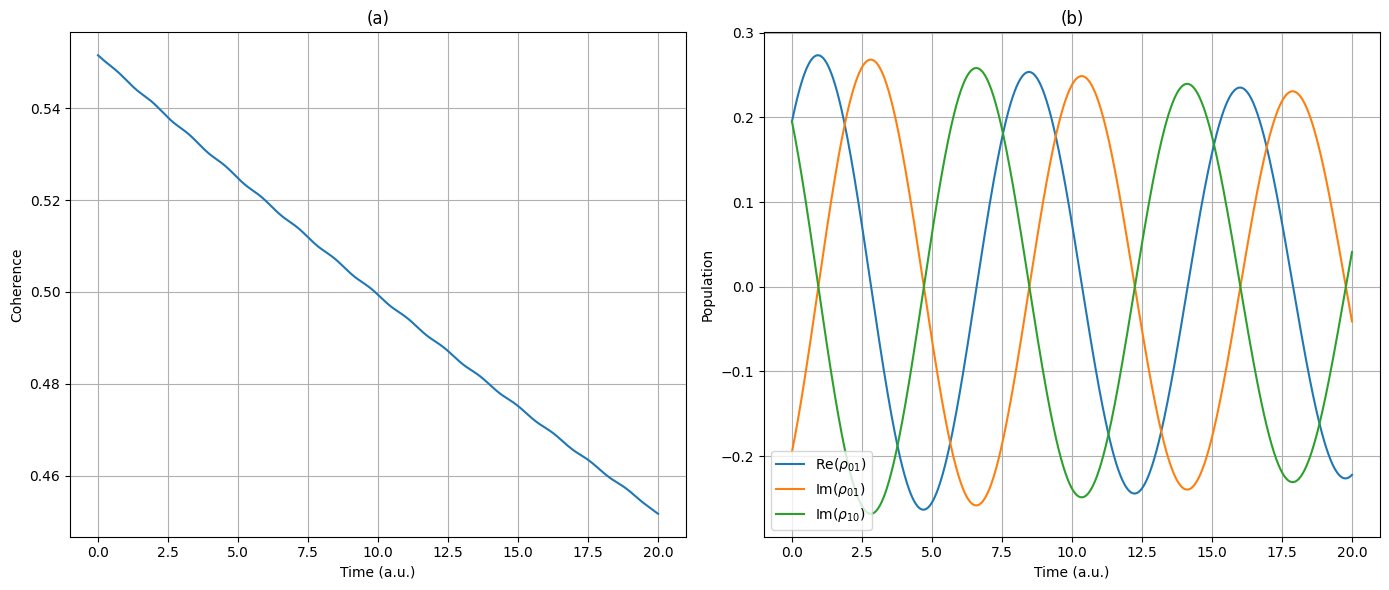}
    \caption{(a) Coherence between states \( |0\rangle \) and \( |1\rangle \) (b) Evolution of off-diagonal density matrix elements over time (with no action of control).}
    \label{fig:NC}
\end{figure}
We now consider including the control Hamiltonian with parameters set to control amplitude $u = 0.1$, the phase $\phi_d=\pi/2$, and the driving frequency 
$\omega_d = 0.1$. 
As illustrated in Fig.~\ref{fig:CC}, the coherence between states shows an oscillatory decay over time, indicating that while the control field momentarily preserves coherence, it cannot completely counteract the effects of decoherence and dephasing. The density matrix elements also exhibit oscillations, reflecting the influence of the control field on the population redistribution among the qutrit states. 
The behavior highlights the transient effect of the control field in partially mitigating the loss of coherence.

\begin{figure}
    \centering
    \includegraphics[scale=0.25]{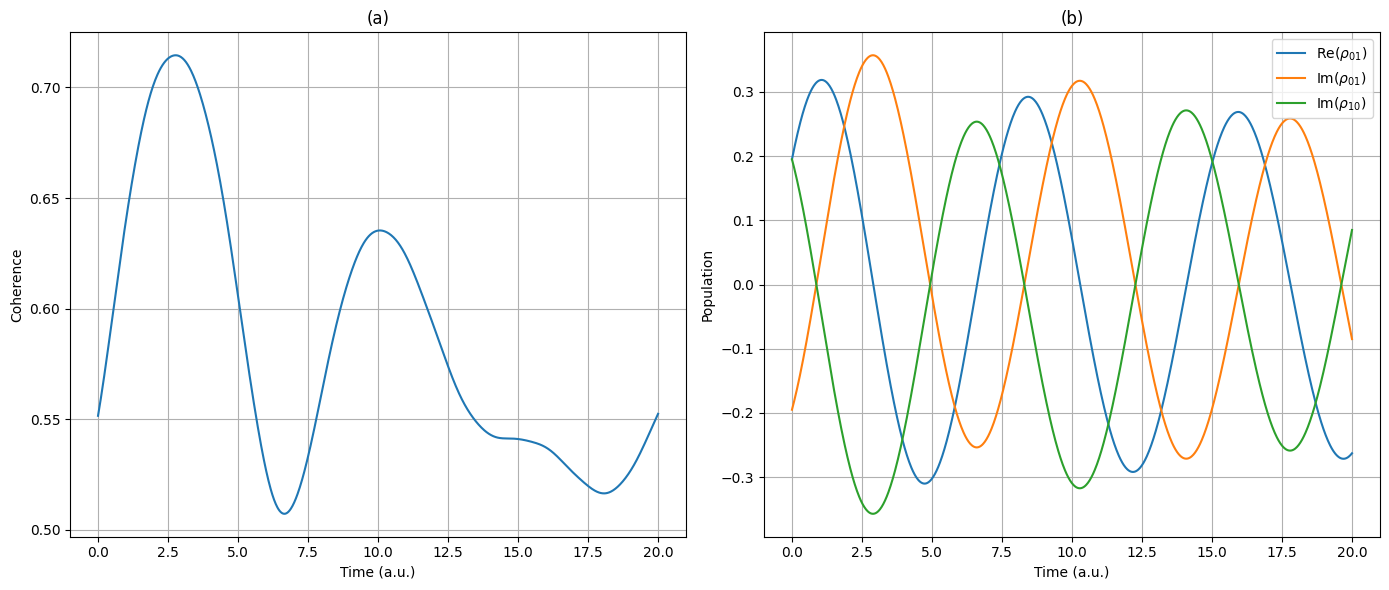}
    \caption{(a) Coherence between states \( |0\rangle \) and \( |1\rangle \) (b) Evolution of density matrix elements over time (with constant control parameters, with dephasing).}
    \label{fig:CC}
\end{figure}

Finally, we demonstrate the results obtained from solving the state constrained optimal control problem for preserving quantum coherence.
The initial coherence is $C(\rho_0)=0.551$, and the lower and upper bounds were set to $\alpha=0.550$ and $\beta=0.553$, respectively. 
To provide a comprehensive view of the quantum system's behavior under optimal control, Fig.~\ref{fig:uopt_analysis} includes four subplots. 
Subplot (a) displays the temporal evolution of the real and imaginary parts of the off-diagonal density matrix element $\rho_{01}$. The controlled oscillations indicate active modulation aimed at preserving coherence despite environmental interactions.
Subplot (b) shows the coherence between the states $|0\rangle$ and $|1\rangle$ over time. The flat behavior of the coherence in this subplot suggests that the applied control manages to maintain a consistent coherence level across the time span.
In subplot (c), the dynamics of the adjoint variables $\pi_{00}$, $\pi_{01}$, $\pi_{12}$, and $\pi_{10}$ are presented. These represent the real components associated with the system's costate evolution. The behavior of these variables provides insight into how the optimization process adjusts the state trajectory to meet the control objectives.
Lastly, subplot (d) captures the control amplitude's time-dependent modulation, which directly impacts the system's evolution. The variations in control illustrate the adaptive strategy employed to counteract the dissipative and dephasing effects, thereby aiding in the coherence preservation. This analysis demonstrates the effectiveness of the control strategy in mitigating coherence loss and guiding the system's dynamics toward desired outcomes. Fig.~\ref{fig:convergence_plot} presents the convergence plot, which depicts the progression of the convergence value over multiple iterations. The decreasing trend observed in the plot demonstrates the effectiveness of the control or optimization strategy being applied, as the convergence value approaches zero or stabilizes, indicating that the system is nearing an optimal or steady state. This behavior is consistent with expectations in optimization scenarios where iterative adjustments progressively refine the solution.
\begin{figure}[t]
    \centering
    \includegraphics[scale=0.25]{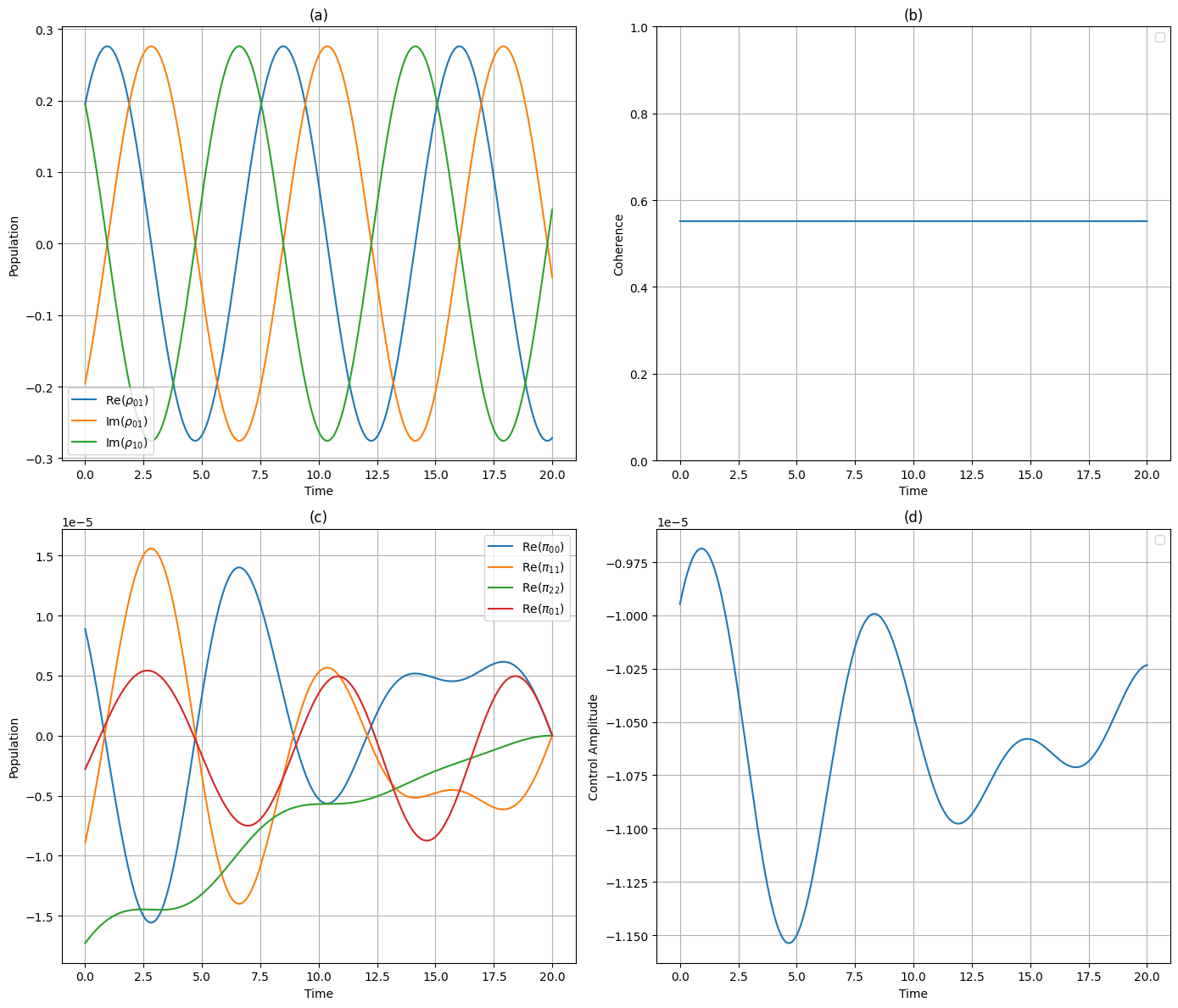}
    \caption{
        (a) Evolution of the real and imaginary components of the density matrix element $\rho_{01}$ under the optimized control strategy. 
        (b) Coherence measured between states $|0\rangle$ and $|1\rangle$ over time. 
        (c) Population dynamics for the adjoint variables $\pi_{00}$, $\pi_{01}$, $\pi_{12}$, and $\pi_{10}$ during the control process. 
        (d) Control amplitude evolution over time illustrating the modulation applied to influence the system's state.
    }
    \label{fig:uopt_analysis}
\end{figure}
\begin{figure}[ht]
    \centering
    \includegraphics[scale=0.25]{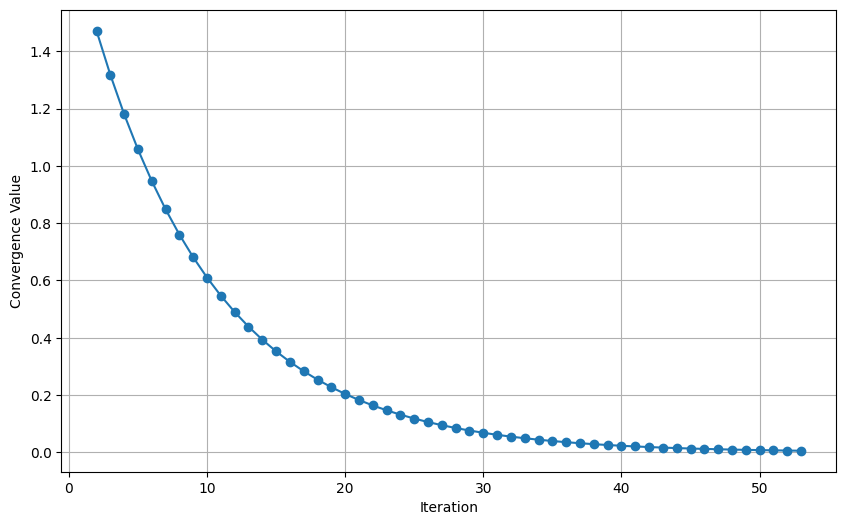}
    \caption{
        Convergence plot showing the evolution of the convergence value as a function of iterations. The plot illustrates how the convergence value decreases progressively. 
    }
    \label{fig:convergence_plot}
\end{figure}

\section{Conclusion}

This paper proposed a minimum energy optimal control strategy for maintaining quantum coherence in a multi level quantum system under dissipative conditions, modeled by the Lindblad master equation. Using Pontryagin’s Minimum Principle with state constraints, our approach effectively balances coherence preservation with minimal energy expenditure in the presence of Markovian decoherence. 

\bibliographystyle{IEEEtran}
\bibliography{IEEEabrv,ref}

\end{document}